\documentclass{article}
\usepackage[utf8]{inputenc}
\usepackage{color}
\usepackage[T1]{fontenc}
\usepackage[normalem]{ulem}
\usepackage[english]{babel}
\usepackage{verbatim}
\usepackage{graphicx}
\usepackage{enumerate}
\usepackage{amsmath,amssymb,amsfonts,amsthm,amscd,mathrsfs}
\usepackage{array}
\usepackage{amsmath,amssymb,graphicx,stmaryrd,enumerate,bbm,alltt}
\usepackage{dsfont}
\usepackage{comment}
\usepackage{mathtools}
\usepackage{bbm}

\usepackage[T1]{fontenc}
\usepackage{babel}
\usepackage[bookmarksopen, bookmarksnumbered]{hyperref}

\usepackage{url}
\usepackage{caption}

\DeclareMathOperator\supp{supp}
\newtheorem{theorem}{Theorem}[section]

\newtheorem{lemma}[theorem]{Lemma}

\newtheorem{prop}[theorem]{Proposition}

\newtheorem{conj}[theorem]{Conjecture}
\theoremstyle{definition}

\newcommand{\Z}{\mathbb Z}
\newcommand{\R}{\mathbb R}

\title{On Support Cardinality for the Discrete Schrödinger Equation}
\author{Linjun Li
\thanks{Department of Mathematics, University of Pennsylvania, Philadelphia, PA.} 
}
\date{}

\begin{document}

\maketitle
\begin{abstract}
How sparse can a nontrivial solution of a discrete Schr\"odinger equation be?  
In this note, we study Dirichlet solutions on a finite $d$-dimensional lattice box, allowing an arbitrary real potential, and measure sparsity by the number of lattice sites at which the solution is nonzero (assuming it is nonzero at the origin).  
Our main result is a dimension-reduction principle: the minimal possible support size cannot decrease when the dimension increases.  
Consequently, any lower bound proved in dimension $d-1$ automatically yields the same lower bound in dimension $d$.  
As an application, we obtain a nearly sharp lower bound in four dimensions, matching the best known two-dimensional constructions up to a logarithmic factor.
\end{abstract}

\section{Introduction}

The discrete Laplacian on $\Z^d$,
\[
(\Delta u)(a)= -2d\,u(a)+\sum_{|a-b|=1}u(b),
\]
is the discrete analog of the usual Laplacian in $\R^{d}$ and also the standard graph Laplacian on the nearest-neighbor lattice.  Adding a potential $V$ leads to the (stationary) \emph{discrete Schr\"odinger equation}
\begin{equation}\label{eq:Sch_intro}
-\Delta u + V u = 0,
\end{equation}
a basic model in discrete potential theory and in the tight-binding description of quantum systems.  When the domain is a finite box and one imposes Dirichlet boundary conditions, \eqref{eq:Sch_intro} becomes a concrete linear system whose solutions can be studied by elementary means; nevertheless, it hides genuinely subtle phenomena that do \emph{not} occur in the continuum.

A guiding principle for solutions of the Schr\"odinger equation on $\R^d$ is \emph{unique continuation}: a nontrivial solution cannot vanish on an open set.  On the lattice, however, solutions may be supported on lower-dimensional subsets (See Proposition \ref{prop:example}), so one is led to ask for \emph{quantitative} substitutes: how sparse can a nonzero solution's support be?  The present note studies this question by measuring sparsity by the \emph{cardinality of the support}.

\medskip

Fix $N\ge 1$ and write
\[
B_{d}(N)=\{(x_1,\dots,x_d): x_i\in\{-N,-N+1,\dots,N\}\}
\] and 
\[
I_N=\{-N,-N+1,\dots,N\}.
\]
We view $\Delta$ as the Laplacian restricted to $B_d(N)$ with \emph{zero boundary condition} on
$\partial B_d(N)=B_d(N+1)\setminus B_d(N)$. For any function $u$ on $B_{d}(N)$, denote the support of $u$ by $\supp(u)=\{x \in B_{d}(N): u(x) \neq 0\}$. 
Define $S_d(N)$ as the minimal possible cardinality of $\supp(u)$ among all functions $u$ on $B_d(N)$ satisfying:
\begin{itemize}
\item there exists a real-valued function $V$ on $B_d(N)$ such that $u$ solves $-\Delta u +Vu=0$ on $B_d(N)$ with zero boundary condition;
\item $u(\mathbf{0})\neq 0$.
\end{itemize}
In other words, $S_d(N)$ asks: \emph{among all discrete Schr\"odinger equations on the box, how sparse can a nontrivial Dirichlet solution be, if it is required to be nonzero at the origin?}

A simple construction (given in Proposition~\ref{prop:example}) shows that sparsity can indeed occur in high dimensions: one can build solutions supported on about $(2N+1)^{d/2}$ points when $d$ is even, and about $(2N+1)^{(d+1)/2}$ points when $d$ is odd.  This raises a natural question, reminiscent of uncertainty principles and unique continuation: \emph{is this exponent essentially optimal, or can one force solutions to occupy a strictly smaller fraction of the box?}

\medskip

Our main result is a monotonicity principle in the dimension.

\begin{theorem}\label{thm:main}
For any $d\ge 2$ and $N\ge 1$, one has
\[
S_{d}(N)\ge S_{d-1}(N).
\]
\end{theorem}

Thus, lower bounds in a given dimension automatically propagate to all higher dimensions. 
For $d=1$, it is easy to see $S_{1}(N)=\Omega(N)$: $\supp(u)$ cannot contain two consecutive zeros without forcing $u$ to vanish on the entire interval $B_{1}(N)$ by iterating the equation, contradicting $u(\mathbf 0)\ne 0$. For $d=2$, the same idea implies that a nonzero solution $u$ cannot vanish on two consecutive lines in the box; thus $S_{2}(N)=\Omega(N)$. For $d=3$, Li and Zhang leveraged previous work of Buhovsky--Logunov--Malinnikova--Sodin \cite{buhovsky2017discrete} and proved that
\begin{equation}\label{eq:3d}
    S_{3}(N)\ \ge\ c\,\frac{N^{2}}{\log N}
\end{equation}
as an implication of Theorem~5.1 in \cite{li2022anderson}. Thus, combined with examples in Proposition \ref{prop:example}, for $d=1,2,3$, the exponent-optimal bound of $S_{d}(N)$ is obtained.
Combining \eqref{eq:3d} with our Theorem~\ref{thm:main} yields the same lower bound in $d=4$:
\[
S_{4}(N)\ \ge\ c\,\frac{N^{2}}{\log N}.
\]
Together with the explicit two-dimensional support construction in Proposition~\ref{prop:example}, this gives an exponent-sharp estimate in four dimensions up to a logarithmic factor:
\[
(2N+1)^2 \ \ge\ S_4(N)\ \ge\ c\,\frac{N^2}{\log N}.
\]

Exponent-sharp lower bounds in higher dimensions remain unknown, and one may conjecture the following.
\begin{conj}
For $d \geq 1$, we have
    $S_{d}(N) = \Omega(N^{\lceil{d/2}\rceil})$.
\end{conj}
It is also illuminating to compare with the \emph{whole-space} problem on $\Z^d$. In that setting, one may ask for the (fractal) \emph{dimension} of the support of a nontrivial solution to \eqref{eq:Sch_intro}.  Krymskii \cite{kry2024discrete} proved that any nonzero solution on $\Z^d$ must have support-dimension at least $\log_2(d)-7$, and also constructed $\Z_2$-valued harmonic examples with dimension about $\log_2(d)+1$. We note that the $\Z_2$-valued example is very different from the real-valued case because the real-valued function can exhibit polynomial structure along tilted lines in $\Z^{d}$ as shown in \cite{li2022anderson,buhovsky2017discrete}.

The motivation of obtaining a lower bound for support (or in general, small-value set) of solution of discrete Schr\"odinger equation is highly related to Anderson localization under Bernoulli potential, dating back to the work of Bourgain--Kenig \cite{bourgain2005localization}: in the continuum, lower bounds preventing an eigenfunction from being ``too small on too large a set'' are closely intertwined with Landis-type questions on how fast a \emph{real-valued} solution of $\Delta u+Vu=0$ can decay at infinity (and with how nodal geometry constrains such decay); see, for instance, the recent survey on Landis' conjecture and related unique continuation phenomena \cite{landisSurvey2024} and the breakthrough progress of Logunov-Malinnikova-Nadirashvili-Nazarov establishing a near-sharp Landis bound for $d=2$ \cite{logunov2025landis}.
On the lattice, quantitative unique continuation in the strongest continuum form is false, but a remarkable substitute survives in $d=2$ thanks to rigid two-dimensional structure. Roughly speaking, Buhovsky--Logunov--Malinnikova--Sodin \cite{buhovsky2017discrete} proved that a discrete harmonic function on $\mathbb Z^2$ that is bounded on a $(1-\varepsilon)$-portion of every sufficiently large square must be constant. 

Later it was found that randomness of the potential $V$ will increase the support dimension.
Building on previous result of Buhovsky--Logunov--Malinnikova--Sodin, Ding and Smart proved a \emph{random} quantitative unique continuation statement for eigenfunctions of \eqref{eq:Sch_intro}: Fix a small interval and let $V$ be an i.i.d. Bernoulli potential. Then with high probability, an eigenfunction of \eqref{eq:Sch_intro} with eigenvalue lying in this interval must occupy a set of sites of size with $\frac{3}{2} - \epsilon$ dimension, in particular, it cannot concentrate on a 1-dimensional sparse set \cite{ding2020localization}.
Later on, for the same random case, the above lower bound of support size is improved to 2-dimensional \cite{li2022ab2d}. For dimension 3, near 2-dimensional lower bound of support \eqref{eq:3d} was proved in \cite{li2022anderson} for any potential. It is still unknown that if the potential is chosen i.i.d. from Bernoulli potential, whether the support of eigenfunctions will be 3-dimensional with high probability. 
For \cite{bourgain2005localization,ding2020localization,li2022ab2d,li2022anderson}, the lower bound for size of small-value set was a key to the proof of localization with Bernoulli potential. For the background of Anderson-Bernoulli model, we refer to \cite{carmona1987bernoulli,germinet2001bootstrap,germinet2011comprehensive} for historical results. In parallel, quantitative transport and diffusion questions in other related disordered models (notably random band matrices and kinetic models) have also advanced rapidly recently, with rigorous results on localization and delocalization \cite{erdos2013rbm,bourgade2020rbm1,yangyauyin2022highd2,xu2024rbm,elboim2025mirrorhighd,li2021manhattan,li2021lorentzcyl,yauyin2025}.

\section{Dimension Reduction}
 
For any solution $u$, we prove that there is a function $\mathbf{F}(u)$ which solves the equation in the $(d-1)$-dimensional box $B_{d-1}(N)$ and $|\{x \in B_{d}(N): u(x) \neq 0\}| \geq |\{x \in B_{d-1}(N): \mathbf{F}(u)(x) \neq 0\}|$.

Indeed, think of $B_{d}(N)$ as a fiber of $B_{d-1}(N)$, i.e. $B_{d}(N) = B_{d-1}(N) \times I_{N}$. For any $x \in B_{d-1}(N)$ with $x=(x_1,...,x_{d-1})$, let $u_{x}$ denote the function on $I_{N}$ such that $u_{x}(y) = u(x_1,...,x_{d-1}, y)$ for $y \in I_{N}$.
\begin{lemma}\label{lem:vector}
There exists a vector $a \in \R^{2N+1}$ such that, for any $x \in B_{d-1}(N)$, we have $u_{x} \neq 0$ if and only if $ \langle a, u_{x} \rangle \neq 0$.
\end{lemma}
\begin{proof}
Let $m:=2N+1$, so that each fiber $u_x$ is naturally identified with a vector in $\R^{m}$ via
\[
u_x \equiv \bigl(u_x(-N),u_x(-N+1),\dots,u_x(N)\bigr)\in\R^{m}.
\]
Consider the finite index set
\[
\mathcal X:=B_{d-1}(N),
\]
and let
\[
\mathcal S:=\{x\in\mathcal X:\ u_x\neq 0\}.
\]
If $\mathcal S=\varnothing$, then $u_x=0$ for every $x\in\mathcal X$ and the statement holds for any choice of $a$ (for instance $a=e_1$), since then $\langle a,u_x\rangle=0$ for all $x$.

Assume henceforth that $\mathcal S\neq\varnothing$. For each $x\in\mathcal S$, define the hyperplane
\[
H_x:=\{a\in\R^{m}:\ \langle a,u_x\rangle=0\}.
\]
Because $u_x\neq 0$, the map $a\mapsto \langle a,u_x\rangle$ is a nonzero linear functional on $\R^{m}$, and therefore $H_x$ is a proper linear subspace of codimension $1$ in $\R^{m}$.

Since $\mathcal S$ is finite (indeed $\mathcal X$ is finite), the union
\[
H:=\bigcup_{x\in\mathcal S} H_x
\]
is a finite union of codimension $1$ subspaces. Therefore $\R^{m}\setminus H$ is nonempty. Choose any
\[
a\in \R^{m}\setminus H.
\]
Then for every $x\in\mathcal S$ we have $a\notin H_x$, i.e.\ $\langle a,u_x\rangle\neq 0$. On the other hand, if $x\in\mathcal X\setminus\mathcal S$ then $u_x=0$ and hence $\langle a,u_x\rangle=0$.

Combining the two cases, for every $x\in B_{d-1}(N)$ we have
\[
u_x\neq 0 \quad\Longleftrightarrow\quad \langle a,u_x\rangle\neq 0,
\]
which proves the lemma.
\end{proof}

Assume $d \geq 2$. Pick a vector $a$ that satisfies Lemma \ref{lem:vector} and denote $\mathbf{F}(u)(x) = \langle a, u_{x}\rangle$ for any $x \in B_{d-1}(N)$. 
\begin{lemma}\label{lem:reduce}
There exists a real-valued function $V_{1}$ on $B_{d-1}(N)$ such that $-\Delta \mathbf{F}(u) + V_{1} \mathbf{F}(u) = 0$ on $B_{d-1}(N)$ with the zero boundary condition.
\end{lemma}
\begin{proof}
Write $m:=2N+1$ and identify each fiber $u_x$ with the vector
\[
u_x=\bigl(u(x,-N),u(x,-N+1),\dots,u(x,N)\bigr)\in\R^{m}.
\]
Fix a vector $a\in\R^{m}$ as in Lemma~\ref{lem:vector} (applied on the finite set of fibers $\{u_x:x\in B_{d-1}(N)\}$), and define
\[
\mathbf F(u)(x):=\langle a,u_x\rangle,\qquad x\in B_{d-1}(N).
\]

\medskip
\noindent\textbf{Step 1: apply $\langle a,\cdot\rangle$ to the equation fiberwise.}
For each interior point $(x,y)\in B_{d-1}(N)\times I_N=B_d(N)$, the equation $-\Delta u+Vu=0$ reads
\begin{align*}
0
&= -\Bigl(-2d\,u(x,y)+\sum_{i=1}^{d-1}\bigl(u(x+\mathbf e_i,y)+u(x-\mathbf e_i,y)\bigr)
      +u(x,y+1)+u(x,y-1)\Bigr) \\
&\qquad\qquad +\,V(x,y)\,u(x,y).
\end{align*}
Multiply this identity by the coordinate $a_y$ (the component of $a$ indexed by $y\in I_N$) and sum over $y\in I_N$.
Using linearity and the definition of $\mathbf F(u)$, we obtain for every $x\in B_{d-1}(N)$:
\begin{equation}\label{eq:key_identity}
-\Delta_{d-1}\mathbf F(u)(x)
\;+\;
\sum_{y\in I_N} a_y\Bigl(-\Delta^{(1)} u_x(y)+V(x,y)\,u_x(y)\Bigr)=0.
\end{equation}
Here $\Delta_{d-1}$ is the $(d-1)$-dimensional lattice Laplacian on $B_{d-1}(N)$ with zero boundary condition, and
\[
\Delta^{(1)} f(y):=-2f(y)+f(y-1)+f(y+1)
\]
is the one-dimensional lattice Laplacian in the last coordinate, with the Dirichlet convention $f(y)=0$ for $y\notin I_N$
(which matches the boundary condition for $u$ at $y=\pm(N+1)$).
To justify \eqref{eq:key_identity}, note that the $(d-1)$-dimensional part of the Laplacian satisfies
\[
\sum_{y\in I_N} a_y\Bigl(-2(d-1)\,u(x,y)+\sum_{i=1}^{d-1}\bigl(u(x+\mathbf e_i,y)+u(x-\mathbf e_i,y)\bigr)\Bigr)
=
\Delta_{d-1}\mathbf F(u)(x),
\]
since $\mathbf F(u)(x\pm\mathbf e_i)=\sum_{y\in I_N} a_y\,u(x\pm\mathbf e_i,y)$ and boundary neighbors are interpreted as $0$.

Define, for $x\in B_{d-1}(N)$,
\[
T(x):=\sum_{y\in I_N} a_y\Bigl(-\Delta^{(1)} u_x(y)+V(x,y)\,u_x(y)\Bigr)
=\Bigl\langle a,\,-\Delta^{(1)}u_x+ (V_x\odot u_x)\Bigr\rangle,
\]
where $V_x\in\R^m$ denotes the vector $y\mapsto V(x,y)$ and $\odot$ is pointwise product.
Then \eqref{eq:key_identity} becomes
\begin{equation}\label{eq:DeltaF_T}
-\Delta_{d-1}\mathbf F(u)(x) + T(x)=0,\qquad x\in B_{d-1}(N).
\end{equation}

\medskip
\noindent\textbf{Step 2: define $V_1$ and verify the Schr\"odinger form.}
By Lemma~\ref{lem:vector}, for each $x\in B_{d-1}(N)$ we have
\[
\mathbf F(u)(x)=\langle a,u_x\rangle=0 \quad\Longleftrightarrow\quad u_x\equiv 0.
\]
In particular, if $\mathbf F(u)(x)=0$ then $u_x\equiv 0$, hence $\Delta^{(1)}u_x\equiv 0$ and $V_x\odot u_x\equiv 0$, so $T(x)=0$.
Therefore the ratio $T(x)/\mathbf F(u)(x)$ is well-defined if we interpret it as $0$ at points where $\mathbf F(u)(x)=0$.

Define $V_1:B_{d-1}(N)\to\R$ by
\[
V_1(x):=
\begin{cases}
\dfrac{T(x)}{\mathbf F(u)(x)}, & \text{if }\mathbf F(u)(x)\neq 0,\\[1.2ex]
0, & \text{if }\mathbf F(u)(x)=0.
\end{cases}
\]
This function is real-valued since $a,u,V$ are real-valued.
Moreover, multiplying the definition by $\mathbf F(u)(x)$ yields
\[
V_1(x)\,\mathbf F(u)(x)=T(x)\qquad\text{for all }x\in B_{d-1}(N),
\]
because when $\mathbf F(u)(x)=0$ we have $T(x)=0$ as shown above.
Substituting into \eqref{eq:DeltaF_T}, we conclude that for every $x\in B_{d-1}(N)$,
\[
-\Delta_{d-1}\mathbf F(u)(x) + V_1(x)\,\mathbf F(u)(x)=0.
\]
This is exactly the statement that $\mathbf F(u)$ solves
\[
-\Delta \mathbf F(u) + V_1\,\mathbf F(u)=0
\]
on $B_{d-1}(N)$ under the zero boundary condition. Finally, $F(u)(\mathbf{0}) \neq 0$ because $u(\mathbf{0}) \neq 0$ and Lemma \ref{lem:vector}.
\end{proof}

By Lemma \ref{lem:vector}, we have $|\{ x\in B_{d-1}(N) : \mathbf{F}(u)(x) \neq 0 \}| \leq |\{ x\in B_{d-1}(N) : u_{x} \neq 0 \}| \leq |\{ x\in B_{d}(N) : u(x) \neq 0 \}|$. Thus, we get $S_{d}(N) \geq S_{d-1}(N)$ which proves Theorem \ref{thm:main}.

Finally, we give a construction of solutions of \eqref{eq:Sch_intro} with low-dimensional support.
\begin{prop}\label{prop:example}
    $S_{d}(N)\leq (2N+1)^{d/2}$ for even dimensions, and $S_{d}(N) \leq (2N+1)^{(d+1)/2}$ for odd dimensions.
\end{prop}
\begin{proof}
We exhibit, for each $d$ and $N$, a function $u$ on $B_d(N)$ with $u(\mathbf 0)\neq 0$ and a real potential $V$ on $B_d(N)$ such that
\[
-\Delta u + Vu = 0
\quad\text{on } B_d(N)
\]
(with the Dirichlet convention $u\equiv 0$ on $\partial B_d(N)=B_d(N+1)\setminus B_d(N)$).
Then $S_d(N)$ is at most the cardinality of the support of $u$.

Recall that for $a\in B_d(N)$,
\[
(\Delta u)(a)=-2d\,u(a)+\sum_{b\in\Z^d:\,|a-b|=1}u(b),
\]
where $u(b)=0$ if $b\notin B_d(N)$. Hence the equation $-\Delta u+Vu=0$ is equivalent to
\begin{equation}\label{eq:Sch_equiv}
(V(a)+2d)\,u(a)=\sum_{|a-b|=1}u(b),\qquad a\in B_d(N).
\end{equation}

\medskip
\noindent\textbf{Even dimensions.}
Assume $d=2k$ with $k\ge 1$. Define $u:B_d(N)\to\R$ by
\begin{equation}\label{eq:u_even}
u(x_1,\dots,x_{2k})
:=(-1)^{x_1+x_3+\cdots+x_{2k-1}}\prod_{j=1}^{k}\mathbf 1_{\{x_{2j-1}=x_{2j}\}}.
\end{equation}
Then $u(\mathbf 0)=1\neq 0$. Its support is exactly
\[
\{x\in B_{2k}(N): x_{2j-1}=x_{2j}\text{ for all }j=1,\dots,k\}.
\]
Therefore
\[
|\supp(u)|=(2N+1)^k=(2N+1)^{d/2}.
\]

We claim that \eqref{eq:Sch_equiv} holds for the constant potential $V\equiv -2d=-4k$.
Since $V+2d\equiv 0$, equation \eqref{eq:Sch_equiv} reduces to
\begin{equation}\label{eq:neighbor_sum_zero_even}
\sum_{|a-b|=1}u(b)=0,\qquad a\in B_{2k}(N).
\end{equation}
Fix $a=(a_1,\dots,a_{2k})\in B_{2k}(N)$ and set $\delta_j:=a_{2j-1}-a_{2j}$ for $j=1,\dots,k$.

\smallskip
\emph{Case 1: $a\in\supp(u)$ (i.e.\ $\delta_j=0$ for all $j$).}
Any nearest neighbor $b$ of $a$ is obtained by changing exactly one coordinate by $\pm1$.
This necessarily makes exactly one $\delta_j$ equal to $\pm1$, hence $b\notin\supp(u)$ and so $u(b)=0$.
Thus all neighbor values vanish and \eqref{eq:neighbor_sum_zero_even} holds.

\smallskip
\emph{Case 2: $a\notin\supp(u)$.}
If at least two indices $j$ satisfy $\delta_j\neq 0$, then no single-coordinate move can enforce all equalities
$x_{2j-1}=x_{2j}$ simultaneously, so every nearest neighbor $b$ also lies outside $\supp(u)$ and $u(b)=0$.
Hence the neighbor sum is $0$.

It remains to consider the case that exactly one index $j_0$ satisfies $\delta_{j_0}\neq 0$ and all other $\delta_j=0$.
If $|\delta_{j_0}|\ge 2$, then again no single step can make $\delta_{j_0}=0$, so every neighbor $b$ has $u(b)=0$.
If $\delta_{j_0}=1$, write $s:=a_{2j_0}$ so that $a_{2j_0-1}=s+1$.
Then the only neighbors that satisfy all equalities are
\[
b^{(1)}:=a-\mathbf e_{2j_0-1}
\quad\text{and}\quad
b^{(2)}:=a+\mathbf e_{2j_0},
\]
both of which lie in $B_{2k}(N)$ (since $s\in I_N$ and $s+1\in I_N$).
All other neighbors are off-support and contribute $0$.
Moreover, $b^{(1)}$ and $b^{(2)}$ agree in every coordinate except that the odd coordinate $x_{2j_0-1}$
equals $s$ for $b^{(1)}$ and equals $s+1$ for $b^{(2)}$; hence the parity factor in \eqref{eq:u_even} differs by a sign:
\[
u\bigl(b^{(2)}\bigr)=(-1)^{s+1+\sum_{j\ne j_0}a_{2j-1}}
=-(-1)^{s+\sum_{j\ne j_0}a_{2j-1}}
=-u\bigl(b^{(1)}\bigr).
\]
Thus the two nonzero neighbor contributions cancel, and the total neighbor sum is $0$.
The case $\delta_{j_0}=-1$ is analogous (with the roles of $2j_0-1$ and $2j_0$ swapped).
This proves \eqref{eq:neighbor_sum_zero_even}, and hence \eqref{eq:Sch_equiv} holds with $V\equiv -2d$.

Therefore, for even $d$, there exists a solution-potential pair $(u,V)$ with $|\supp(u)|=(2N+1)^{d/2}$, so
$S_d(N)\le (2N+1)^{d/2}$.

\medskip
\noindent\textbf{Odd dimensions.}
For $d=1$, the proposition is trivial because $|B_{1}(N)| = 2N+1$. 
Assume $d=2k+1$ with $k\ge 1$. Define $u:B_{2k+1}(N)\to\R$ by
\begin{equation}\label{eq:u_odd}
u(x_1,\dots,x_{2k},x_{2k+1})
:=(-1)^{x_1+x_3+\cdots+x_{2k-1}}\prod_{j=1}^{k}\mathbf 1_{\{x_{2j-1}=x_{2j}\}},
\end{equation}
i.e.\ the same pairing constraint on the first $2k$ coordinates, and no constraint on the last coordinate.
Again $u(\mathbf 0)=1\neq 0$, and now the free coordinates are $x_2,x_4,\dots,x_{2k}$ and $x_{2k+1}$, so
\[
|\supp(u)|=(2N+1)^{k+1}=(2N+1)^{(d+1)/2}.
\]

We now specify a potential $V$ for which \eqref{eq:Sch_equiv} holds.
Write $x=(x',x_{2k+1})$ with $x'\in\Z^{2k}$.
Define $V:B_{2k+1}(N)\to\R$ by
\begin{equation}\label{eq:V_odd_def}
V(x):=
\begin{cases}
-2d+2, & \text{if } |x_{2k+1}|<N,\\
-2d+1, & \text{if } |x_{2k+1}|=N.
\end{cases}
\end{equation}

Fix $a=(a',a_{2k+1})\in B_{2k+1}(N)$.

\smallskip
\emph{If $u(a)=0$}, then $a'$ violates at least one pairing equality among the first $2k$ coordinates.
Exactly as in the even-dimensional case, one checks that the sum of neighbor values $\sum_{|a-b|=1}u(b)$ equals $0$:
either there are no support neighbors at all, or there are exactly two support neighbors obtained by fixing a single violated pair with
difference $\pm1$, and their contributions cancel by the same sign argument as above. (Moves in the last coordinate never change $a'$
and thus cannot create support points.) Hence the right-hand side of \eqref{eq:Sch_equiv} is $0$, and the left-hand side is also $0$
because $u(a)=0$, so \eqref{eq:Sch_equiv} holds regardless of $V(a)$.

\smallskip
\emph{If $u(a)\neq 0$}, then $a'$ satisfies all pairing equalities, so any move in one of the first $2k$ coordinates breaks exactly one equality
and lands off-support; thus the only possible nonzero neighbors are those in the last coordinate direction:
\[
a\pm \mathbf e_{2k+1}.
\]
If $|a_{2k+1}|<N$, then both neighbors lie in $B_{2k+1}(N)$ and, since $u$ does not depend on $x_{2k+1}$, we have
$u(a+\mathbf e_{2k+1})=u(a-\mathbf e_{2k+1})=u(a)$. Hence
\[
\sum_{|a-b|=1}u(b)=2u(a).
\]
If $|a_{2k+1}|=N$, then exactly one of these two neighbors remains in $B_{2k+1}(N)$ and the other is outside (hence contributes $0$),
so
\[
\sum_{|a-b|=1}u(b)=u(a).
\]
Inserting these identities into \eqref{eq:Sch_equiv}, we see that \eqref{eq:Sch_equiv} holds precisely when
$V(a)+2d=2$ in the interior ($|a_{2k+1}|<N$) and $V(a)+2d=1$ on the top/bottom faces ($|a_{2k+1}|=N$),
which is exactly the choice \eqref{eq:V_odd_def}. Thus \eqref{eq:Sch_equiv} holds for all $a\in B_{2k+1}(N)$.

Therefore, for any odd $d$, there exists a solution-potential pair $(u,V)$ with $|\supp(u)|=(2N+1)^{(d+1)/2}$, so
$S_d(N)\le (2N+1)^{(d+1)/2}$.

Combining the even and odd cases completes the proof.
\end{proof}
\subsection*{Acknowledgement}
The author thanks Lingfu Zhang for valuable discussions.
\subsection*{Data Availability}
Data sharing not applicable to this article as no datasets were generated or analyzed during
the current study.
\subsection*{Declarations}
On behalf of all authors, the corresponding author states that there is no conflict of interest.
\bibliographystyle{style2}
\bibliography{bib2}

\end{document}